\documentclass[11pt]{article}%
\usepackage{amsfonts}
\usepackage{amsmath}
\usepackage{amssymb}
\usepackage{graphicx}%
\providecommand{\U}[1]{\protect\rule{.1in}{.1in}}
\newtheorem{theorem}{Theorem}
\newtheorem{acknowledgement}[theorem]{Acknowledgement}

\newtheorem{conjecture}[theorem]{Conjecture}
\newtheorem{corollary}[theorem]{Corollary}

\newtheorem{definition}[theorem]{Definition}

\newtheorem{lemma}[theorem]{Lemma}

\newtheorem{proposition}[theorem]{Proposition}
\newtheorem{remark}[theorem]{Remark}

\newenvironment{proof}[1][Proof]{\noindent\textbf{#1.} }{\ \rule{0.5em}{0.5em}}
\begin{document}

\title{Polar Duality Between Pairs of Transversal Lagrangian Planes; Applications to
Uncertainty Principles}
\author{Maurice A. de Gosson\thanks{maurice.de.gosson@univie.ac.at}\\University of Vienna\\Faculty of Mathematics (NuHAG)\\Oskar-Morgenstern-Platz 1\\1090 Vienna AUSTRIA}
\maketitle

\begin{abstract}
We extend the notion of polar duality to pairs $(\ell,\ell^{\prime})$ of
transversal Lagrangian planes in the standard symplectic space. $(\mathbb{R}%
^{2n},\omega)$. This allows us to show that polar duality has a natural
interpretation in terms of symplectic geometry. \ We show that the orthogonal
projections $\Omega_{\ell}$ and $\Omega_{\ell^{\prime}}$ of a centrally
symmetric \ convex body $\Omega$ satisfying $c_{\min}^{\mathrm{lin}}%
(\Omega)\geq\pi$ satisfy the duality relation $\Omega_{\ell}^{o}\subset
\Omega_{\ell^{\prime}}$.

\end{abstract}



\section{Introduction}

Polar duality, in its usual acceptation, associates to a centrally symmetric
convex body $X$ in the Euclidean space $\mathbb{R}^{n}$ another convex body
$X^{o}$ in the dual space $(\mathbb{R}^{n})^{\ast}$ which is, for all
practical purposes, identified with $\mathbb{R}^{n}$ itself: $p^{\prime}\in
X^{o}$ if and only if $p^{\prime}\cdot x\leq1$ for all $x\in X$. The present
paper is based on the observation that if one identifies $\mathbb{R}^{n}$ and
is dual vector space $(\mathbb{R}^{n})^{\ast}$ with the coordinate Lagrangian
planes $\mathbb{\ell}_{X}=\mathbb{R}^{n}\times0$ and $\mathbb{\ell}%
_{P}=0\times\mathbb{R}^{n}$, respectively (for the canonical symplectic
structure $\omega$ on $T^{\ast}\mathbb{R}^{n}$), then the defining relation
$p^{\prime}\cdot x\leq1$ can be rewritten as $\omega(z,z^{\prime})\leq1$ where
$z=(x,0)\in\mathbb{\ell}_{X}$ and $z^{\prime}=(0,p^{\prime})\in\mathbb{\ell
}_{P}$. Thus, polar duality can be reformulated in terms of the symplectic
structure on $T^{\ast}\mathbb{R}^{n}$, and motivates a more general
definition: let $\ell$ and $\ell^{\prime}$ be two transverse Lagrangian planes
in $(T^{\ast}\mathbb{R}^{n},\omega$), and $X_{\ell}\subset\ell$ a centrally
symmetric convex body; the polar dual $X_{\ell^{\prime}}^{o}\subset
\ell^{\prime}$ is then defined as being the set of all $z^{\prime}\in
\ell^{\prime}$ such that $\omega(z,z^{\prime})\leq1$ for all $z\in\ell$. This
highlights the fact tat polar duality is essentially a property of symplectic
geometry. In the present work apply this extended version of polar duality to
the study of the projections of convex sets containing a symplectic ball on
pairs of transverse Lagrangian planes. More precisely:

\begin{itemize}
\item In Section \ref{sec1} we review the geometric tools we will need
(symplectic capacities, standard polar duality, Lagrangian planes), and define
the new concept of Lagrangian polar duality which is central to this work; we
establish the main properties of this duality relate it to symplectic capacities;

\item In Section \ref{sec2} we state and prove the main geometric results of
this paper (Proposition \ref{PropJohn}, Theorems \ref{Thm1} and \ref{Thm2}).
In Proposition \ref{PropJohn} we determine the John ellipsoid of the product
$X\times X^{o}$ and in Theorem \ref{Thm1} we show that if $(\ell,\ell^{\prime
})$ is a pair of transverse Lagrangian planes then the projections $\Pi_{\ell
}:\mathbb{R}^{2n}\longrightarrow\ell$ and $\Pi_{\ell^{\prime}}\mathbb{R}%
^{2n}\longrightarrow\ell$ in the directions $\ell^{\prime}$ and $\ell$ of a
symmetric convex body $\Omega\subset\mathbb{R}^{2n}$ containing a symplectic
ball $S(B^{2n}(1))$ form a Lagrangian polar dual pair. Theorem \ref{Thm2} is a
partial converse to Theorem \ref{Thm1}, highlighting the role played by the
John ellipsoid of a Lagrangian polar dual pair.

\item In Section \ref{secup} we discuss from a new point of view the
indeterminacy principle of quantum mechanics. We justify in Proposition
\ref{PropUP} the following claim: a quantum system localized in a in a
Lagrangian plane $\ell$ cannot be localized in a transverse Lagrangian plane
$\ell^{\prime}$ in a set smaller than its Lagrangian polar dual. This is a
geometric restatement of the uncertainty principle, which has the advantage of
being independent of the choice of any statistical measure.

\item In Section \ref{sechardy} we apply Lagrangian polar duality to a
reformulation of the multi-dimensional Hardy uncertainty principle about the
simultaneous decrease at infinity of a function and its Fourier transform,
thus generalizing previous results of ours.\bigskip
\end{itemize}

\noindent\textbf{Notation}. We denote by $\omega$ the standard symplectic form
on $T^{\ast}\mathbb{R}^{n}\equiv\mathbb{R}^{2n}$: $\omega=\sum_{j=1}^{n}%
dp_{j}\wedge dx_{j}$ (we use the splitting $z=(x,p)$. A diffeomorphism
$f\in\operatorname*{Diff}(n)$ of $\mathbb{R}^{2n}$ is called a
symplectomorphism of $(\mathbb{R}^{2n},\omega)$ if $f^{\ast}\omega=\omega$.
Symplectomorphisms form a subgroup$\ \operatorname*{Symp}(n)$ of
$\operatorname*{Diff}(n)$ equipped with the usual composition operation
\cite{Leonid}. Linear symplectomorphisms form a subgroup $\operatorname*{Sp}%
(n)$ of $\operatorname*{Symp}(n)$. We denote by $\operatorname*{Mp}(n)$ the
unitary representation (metaplectic group) of the double cover of
$\operatorname*{Sp}(n)$. Every $S\in\operatorname*{Sp}(n)$ is the projection
of two operators $\pm\widehat{S}\in\operatorname*{Mp}(n)$.

\section{Prerequisites\label{sec1}}

We will use the following notation for closed balls and symplectic cylinders
in $\mathbb{R}^{2n}$:
\begin{align*}
B^{2n}(R)  &  =\{z\in\mathbb{R}^{2n}:|z|^{2}\leq R^{2}\}\\
Z_{j}^{2n}(R)  &  =\{z\in\mathbb{R}^{2n}:|z_{j}|^{2}\leq R^{2}\}
\end{align*}
where $z_{j}=(x_{j},p_{j})$, $1\leq j\leq n$. Similarly, we denote by
$B^{2n}(z_{0},R)$ (resp. $Z_{j}^{2n}(z_{0},R)$) the ball (\textit{resp.}
cylinder) centered at $z_{0}\in\mathbb{R}^{2n}$. Recall that in view of
Gromov's \cite{gr85} symplectic non-squeezing theorem there exists
$f\in\operatorname*{Symp}(n)$ such that $f(B^{2n}(R))\subset Z_{j}^{2n}(r)$ if
and only if $R\leq r$. We will call \textquotedblleft symplectic
ball\textquotedblright\ (with radius $R$) the image of $B^{2n}(R)$ by
$S\in\operatorname*{Sp}(n)$.

\subsection{Symplectic capacities on the standard symplectic space}

A (normalized) symplectic capacity on $(\mathbb{R}^{2n},\omega)$ associates to
every subset $\Omega\subset\mathbb{R}^{2n}$ a number $c(\Omega)\in
\mathbb{[}0,+\infty\mathbb{]}$ such that the following properties hold
\cite{ekhof1,ekhof2}:

\begin{itemize}
\item \textit{Monotonicity}: If $\Omega\subset\Omega^{\prime}$ then
$c(\Omega)\leq c(\Omega^{\prime})$;

\item \textit{Conformality}: For every $\lambda\in\mathbb{R}$ we have
$c(\lambda\Omega)=\lambda^{2}c(\Omega)$;

\item \textit{Symplectic invariance}: $c(f(\Omega))=c(\Omega)$ for every
$f\in\operatorname*{Symp}(n)$;

\item \textit{Normalization}: We have, for $1\leq j\leq n$,
\begin{equation}
c(B^{2n}(R))=\pi R^{2}=c(Z_{j}^{2n}(R))~. \label{cbz}%
\end{equation}

\end{itemize}

The number $\pi R^{2}$ is sometimes called the symplectic radius of
$B^{2n}(R)$ (or $Z_{j}^{2n}(R)$). We will avoid this terminology because it is
confusing. Note that the conformality and normalization properties show that
for $n>1$ symplectic capacities are not related to volume; they have the
dimension of an area. This is exemplified by the Hofer--Zehnder capacity
\cite{HZ} which has the property that when $\Omega$ is a compact convex set in
$(\mathbb{R}^{2n},\omega)$ with smooth boundary $\partial\Omega$ then
\begin{equation}
c_{\mathrm{HZ}\ }(\Omega)=\int_{\gamma_{\min}}p_{1}dx_{1}+\cdot\cdot
\cdot+p_{n}dx_{n} \label{HZ}%
\end{equation}
where $\gamma_{\min}$ is the shortest (positively oriented) Hamiltonian
periodic orbit carried by $\partial\Omega$.

It follows from Gromov's symplectic non-squeezing theorem \cite{gr85} that the
formulas%
\begin{subequations}
\begin{align}
c_{\min}(\Omega)  &  =\sup_{f\in\operatorname*{Symp}(n)}\{\pi R^{2}%
:f(B^{2n}(R))\subset\Omega\}\label{cmin}\\
c_{\max}(\Omega)  &  =\inf_{f\in\operatorname*{Symp}(n)}\{\pi R^{2}%
:f(\Omega)\subset Z_{j}^{2n}(R) \label{cmax}%
\end{align}
define two symplectic capacities ($c_{\min}$ is sometimes called
\textquotedblleft Gromov's width\textquotedblright\ while $c_{\max}$ is
referred to as the \textquotedblleft cylindrical capacity\textquotedblright).
The notation is motivated by the fact that for every symplectic capacity on
$(\mathbb{R}^{2n},\omega)$ we have $c_{\min}\leq c\leq c_{\max}$.

Linear symplectic capacities $c^{\mathrm{lin}}$ are defined as above,
replacing everywhere the group $\operatorname*{Symp}(n)$ of symplectomorphisms
with the affine symplectic group $\operatorname*{ASp}(n)$ generated by the
affine symplectic transformations $ST(z_{0})=T(Sz_{0})S$\ where $z_{0}%
\in\mathbb{R}^{2n}$ and $S\in\operatorname*{Sp}(n)$ ($T(z_{0})$ is the
translation $z\longmapsto z+z_{0}$). One defines the minimum and maximum
linear symplectic capacities by
\end{subequations}
\begin{subequations}
\begin{align}
c_{\min}^{\mathrm{lin}}(\Omega)  &  =\sup_{\substack{S\in\operatorname*{Sp}%
(n)\\z_{0}\in\mathbb{R}^{2n}}}\{\pi R^{2}:S(B^{2n}(z_{0},R))\subset
\Omega\}\label{clmin}\\
c_{\max}^{\mathrm{lin}}(\Omega)  &  =\inf_{_{\substack{S\in\operatorname*{Sp}%
(n)\\z_{0}\in\mathbb{R}^{2n}}}}\{\pi R^{2}:ST\Omega)\subset Z_{j}^{2n}%
(z_{0},R)\}~ \label{clmax}%
\end{align}
and one has $c_{\min}^{\mathrm{lin}}\leq c^{\mathrm{lin}}\leq$ $c_{\max
}^{\mathrm{lin}}$ for every linear symplectic capacity $c^{\mathrm{lin}}$.

Of particular interest is the symplectic capacity of an ellipsoid. In what
follows we assume that $\Omega$ is an ellipsoid in $\mathbb{R}^{2n}$ given by
\end{subequations}
\[
\Omega=\{z\in\mathbb{R}^{2n}:Mz\cdot z\leq1\}
\]
where $M$ is a positive definite (symmetric) matrix of order $2n$. It is
well-known that for every symplectic capacity $c$ (\textit{resp.} linear
symplectic capacity $c^{\mathrm{lin}}$) we have
\begin{equation}
c(\Omega)=c^{\mathrm{lin}}(\Omega)=\frac{\pi}{\lambda_{\max}^{\omega}}
\label{capellipsoid}%
\end{equation}
where $\lambda_{\max}^{\omega}$ is the largest symplectic eigenvalue of the
matrix $M$ (the symplectic eigenvalues of $M$ are \cite{Birk,HZ} the numbers
$\lambda_{j}^{\omega}>0$ ($1\leq j\leq n$) such that the $\pm i\lambda
_{j}^{\omega}$ are the eigenvalues of the antisymmetric matrix $M^{1/2}%
JM^{1/2}$).

\subsection{Polar duality}

Let $X$ be a closed convex body in $\mathbb{R}_{x}^{n}$ (a convex body in an
Euclidean space is a compact convex set with non-empty interior). We assume in
addition that $X$ contains $0$ in its interior. This is the case if, for
instance, $X$ is centrally symmetric: $X=-X$. By definition The \emph{polar
dual} of $X$ is the subset
\begin{equation}
X^{o}=\{p\in\mathbb{R}^{n}:px\leq1\text{ \textit{for all} }x\in X\}
\label{omo}%
\end{equation}
of the dual space $\mathbb{R}_{p}^{n}\equiv(\mathbb{R}^{n})^{\ast}$ of
$\mathbb{R}_{x}^{n}\equiv\mathbb{R}_{x}^{n}$ (with this notation the phase
space $\mathbb{R}^{2n}$ is identified with the product $T^{\ast}%
(\mathbb{R}^{n})\equiv\mathbb{R}_{x}^{n}\times\mathbb{R}_{p}^{n}$.

It follows from the definition (\ref{omo}) that $X^{o}$ is convex. The
following properties of the polar dual are obvious:
\begin{equation}
\text{\textit{Biduality}: }(X^{o})^{o}=X~; \label{biduality}%
\end{equation}%
\begin{equation}
\text{\textit{Anti-monotonicity: }}X\subset Y\Longrightarrow Y^{o}\subset
X^{o}~; \label{antimonotonicity}%
\end{equation}%
\begin{equation}
\text{\textit{Scaling}: }L\in GL(n,\mathbb{R})\Longrightarrow(LX)^{o}%
=(L^{T})^{-1}X^{o}~. \label{scaling}%
\end{equation}

We also list the following properties of the polar duals of balls and ellipsoids:

Let $B_{X}^{n}(R)$ (\textit{resp}. $B_{P}^{n}(R)$) be the ball $\{x:|x|\leq
R\}$ in $\mathbb{R}_{x}^{n}$ (\textit{resp}. $\{p:|p|\leq R\}$ in
$\mathbb{R}_{p}^{n}$). We have
\begin{equation}
B_{X}^{n}(R)^{o}=B_{P}^{n}(R^{-1})~. \label{BhR}%
\end{equation}
In particular $B_{X}^{n}(1)^{o}=B_{P}^{n}(1)$. When the context is clear we
will write $B_{X}^{n}(R)=B_{P}^{n}(R)=B^{n}(R)$, \textit{etc}.

Let $L\in GL(n,\mathbb{R})$. In view of (\ref{scaling}) we have, for $L\in
GL(n,\mathbb{R})$,
\begin{equation}
(L(B^{n}(R)))^{o}=(L^{T})^{-1}B^{n}(R^{-1}) \label{BHL}%
\end{equation}
hence, if $A=A^{T}\in GL(n,\mathbb{R})$%
\begin{equation}
\{x\in\mathbb{R}_{x}^{n}:Ax\cdot x\leq1\}^{o}=\{p\in\mathbb{R}_{p}^{n}%
:A^{-1}p\cdot p\cdot\leq1\} \label{dualellh}%
\end{equation}
since the left-hand side is $(A^{-1/2}(B^{n}(1)))^{o}=A^{1/2}(B^{n}(1))$.

\begin{definition}
A pair $(X,P)\ $of centrally symmetric convex bodies $X\subset\mathbb{R}%
_{x}^{n}$ and $P\subset\mathbb{R}_{p}^{n}$ is called a \textquotedblleft dual
pair\textquotedblright\ if we have $X^{o}\subset P$ (or, equivalently,
$P^{o}\subset X$). If $X^{o}=P$ we say that $(X,P)$ is an exact dual pair.
\end{definition}

The following elementary result is straightforward:

\begin{lemma}
\label{LemmaYQ}\textit{Let} $(X,P)$ \textit{be a dual pair and} $Y,Q$
\textit{be symmetric} \textit{convex bodies such that }$X\subset Y$ and
$P\subset Q$\textit{. Then} $(Y,Q)$ \textit{is also a dual pair.}
\end{lemma}

\begin{proof}
It follows from the chain of inclusions $Y^{o}\subset X^{o}\subset P\subset Q$
using the anti-monotonicity of the passage to the dual.
\end{proof}

\begin{proposition}
\label{propell}\textit{The ellipsoids} $X=\{x:Ax^{2}\leq1\}$ \textit{and
}$P=\{p:Bp^{2}\leq1\}$ ($A,B>0$) \textit{form a dual pair if and only if
}$A\leq B^{-1}$ (or, equivalently, $AB\leq I_{n\times n}$), \textit{and
}$(X,P)$ is an exact dual pair \textit{if and only if} $AB=I_{n\times n}$.
\end{proposition}

\begin{proof}
We have $X=A^{-1/2}(B^{n}(1))$ and $P=B^{-1/2}(B^{n}(1))$; in view of
(\ref{BHL}) \ the relation $X^{o}\subset P$ is thus equivalent to $A^{1/2}\leq
B^{-1/2}$ (with equality if and only if $X^{o}=P$).
\end{proof}

In \cite{arkaos13}, Remark 4.2, Artstein-Avidan \textit{et} \textit{al }prove
that\textit{ }if $(X,P)$ is an arbitrary pair of centrally symmetric convey
bodies $X\subset\mathbb{R}_{x}^{n}$ and $P\subset\mathbb{R}_{p}^{n}$ then we
have
\begin{equation}
c_{\mathrm{HZ}\ }(X\times P)=c_{\max}(X\times P)=4\sup\{\lambda>0:\lambda
X^{o}\subset P\}~.\label{yaron1}%
\end{equation}
It follows that if $(X,P)$ is a dual pair, then%
\begin{equation}
c_{\mathrm{HZ}\ }(X\times P)=c_{\max}(X\times P)\geq4\label{yaron2}%
\end{equation}
and, in particular,
\begin{equation}
c_{\mathrm{HZ}\ }(X\times P)=c_{\max}(X\times P)=4~.\label{yaron3}%
\end{equation}

\subsection{Lagrangian polar duality}

Let $(\mathbb{R}^{2n},\omega)$ be the standard symplectic space. We denote by
$\Lambda(n)$ the Lagrangian Grassmannian of $(\mathbb{R}^{2n},\omega)$: its
elements (the \textquotedblleft Lagrangian planes\textquotedblright) are the
$n$-dimensional subspaces $\ell$ of $\mathbb{R}^{2n}$ on which $\omega$
vanishes identically. We denote by $\operatorname*{Sp}(n)$ the standard
symplectic group. It consists of all automorphisms of $\mathbb{R}^{2n}$
preserving the symplectic form $\omega$. It follows that $\operatorname*{Sp}%
(n)$ is a closed subgroup of $GL(2n,\mathbb{R})$ and hence a classical Lie
group. The symplectic group acts transitively on the Lagrangian Grassmannian:
for every pair $(\ell,\ell^{\prime})\in\Lambda(n)^{2}$ there exists
$S\in\operatorname*{Sp}(n)$ such that $\ell^{\prime}=S\ell$. This is most
easily seen choosing two bases $\{e_{1},...,e_{n}\}$ and $\{e_{1}^{\prime
},...,e_{n}^{\prime}\}$ of $\ell$ and $\ell^{\prime}$, respectively and
completing these bases to symplectic bases $\{e_{1},...,e_{n};f_{1}%
,...,f_{n}\}$ and $\{e_{1}^{\prime},...,e_{n}^{\prime};f_{1}^{\prime
},...,f_{n}^{\prime}\}$ of $(\mathbb{R}^{2n},\omega)$. The automorphism $S$ of
$\mathbb{R}^{2n}$ taking the first basis to the second is in
$\operatorname*{Sp}(n)$ and we have $\ell^{\prime}=S\ell$ by construction.

Let $U(n)$ be the subgroup of $\operatorname*{Sp}(n)$ consisting of symplectic
rotations:
\[
U(n)=\operatorname*{Sp}(n)\cap O(2n,\mathbb{R}).
\]
The monomorphism $A+iB\longmapsto%
\begin{pmatrix}
A & -B\\
B & A
\end{pmatrix}
$ identifies the unitary group $U(n,\mathbb{C})$ with $U(n)$. An argument
similar to the one above shows that the action $U(n)\times\Lambda
(n)\longrightarrow\Lambda(n)$ is transitive as well; it allows to identify
$\Lambda(n)$ topologically with the homogeneous space $U(n)/O(n)$ where $O(n)$
is the subgroup of $U(n)$ consisting of the mappings $(x,p)\longmapsto
(Ax,Ap)$, $A\in O(n,\mathbb{R})$. It implies that:

\begin{lemma}
\label{Lemparameter}Let $\ell$ be a $n$-dimensional subspace of $\mathbb{R}%
^{2n}$. It is a Lagrangian plane in $(\mathbb{R}^{2n},\omega)$ if and only if
there exists real $n\times n$ matrices $A$ and $B$ satisfying $A^{T}B=B^{T}A$
and $A^{T}A+B^{T}B=I_{n\times n}$ (or $AB^{T}=BA^{T}$ and $AA^{T}%
+BB^{T}=I_{n\times n}$) such that $(x,p)\in\ell$ if and only $Ax+Bp=0$.
\end{lemma}

\begin{proof}
It follows from the transitivity of the action of $U(n)$ on $\Lambda(n)$ which
allows to parametrize $\ell$ by $x=B^{T}u$ and $p=-A^{T}u$, $u\in
\mathbb{R}^{n}$ (observe that $AA^{T}+BB^{T}=I_{n\times n}$ implies that
$\operatorname*{rank}(A,B)=n$).
\end{proof}

The following classical result from symplectic geometry \cite{Birk} will be
used several times. We will say that two Lagrangian planes $\ell$ and
$\ell^{\prime}$ are transverse if $\ell\cap\ell^{\prime}=0$; equivalently
$\ell\oplus\ell^{\prime}=\mathbb{R}^{2n}$.

\begin{lemma}
\label{Lemma1}The symplectic group $\operatorname*{Sp}(n)$ acts transitively
on the set $\Lambda_{0}(n)^{2}$ of all transverse Lagrangian planes in
$(\mathbb{R}^{2n},\omega)$: if $(\ell_{1},\ell_{1}^{\prime})\in\Lambda(n)^{2}$
and $(\ell_{2},\ell_{2}^{\prime})\in\Lambda(n)^{2}$ are such that $\ell
_{1}\cap\ell_{1}^{\prime}=\ell_{2}\cap\ell_{2}^{\prime}=0$ then there exists
$S\in\operatorname*{Sp}(n)$ such that $(\ell_{1},\ell_{1}^{\prime})=S(\ell
_{2},\ell_{2}^{\prime})$.
\end{lemma}

\begin{proof}
Choose a basis $\{e_{11},...,e_{1n}\}$ of $\ell_{1}$ and a basis
$\{f_{11},...,f_{1n}\}$ of $\ell_{1}^{\prime}$ such that $\{e_{1i}%
,f_{1j}\}_{1\leq i,j\leq n}$ is a symplectic basis of $(\mathbb{R}_{z}%
^{2n},\sigma)$. Similarly choose bases of $\ell_{2}$ and $\ell_{2}^{\prime}$
whose union $\{e_{2i},f_{2j}\}_{1\leq i,j\leq n}$ is also a symplectic basis.
The linear mapping $S:\mathbb{R}_{z}^{2n}\longrightarrow\mathbb{R}_{z}^{2n}$
defined by $S(e_{1i})=e_{2i}$ and $S(f_{1i})=f_{2i}$ for $1\leq i\leq n$ is in
$\operatorname*{Sp}(n)$ and $(\ell_{2},\ell_{2}^{\prime})=(S\ell_{1},S\ell
_{1}^{\prime})$.
\end{proof}

We define Lagrangian polar duality:

\begin{definition}
Let $\ell$ and $\ell^{\prime}$ be two transversal Lagrangian planes, and
$X_{\ell}$ a centrally symmetric convex body in $\ell$. The Lagrangian polar
dual $X_{\ell^{\prime}}^{o}$ of $X_{\ell}$ in $\ell^{\prime}$ is the subset of
$\ell^{\prime}$ consisting of all $z^{\prime}\in\ell^{\prime}$ such that
\begin{equation}
\omega(z,z^{\prime})\leq1\text{ \ for all \ }z^{\prime}\in X~.\label{ozz}%
\end{equation}

\end{definition}

The Lagrangian polar dual $X_{\ell^{\prime}}^{o}$ is also centrally symmetric.
In the particular case $\ell=\ell_{X}=\mathbb{R}_{x}^{n}\times0$ and
$\ell^{\prime}=\ell_{P}=0\times\mathbb{R}_{p}^{n}$ this reduces to ordinary
polarity as studied above: for $z=(x,p)\in X_{\ell}$ and $z^{\prime
}=(x^{\prime},p^{\prime})\in X_{\ell^{\prime}}^{o}$ we have indeed
\[
\omega(z,z^{\prime})=\omega((0,p;x^{\prime},0)=px^{\prime}~.
\]

Using the transitivity of the action $\operatorname*{Sp}(n)\times
\Lambda(n)\longrightarrow\Lambda(n)$ one has the following symplectic
covariance of which relates Lagrangian polar duality to ordinary polar duality:

\begin{proposition}
\label{propsycogeom}Let $S\in\operatorname*{Sp}(n)$ be such that $(\ell
,\ell^{\prime})=S(\ell_{X},\ell_{P})$. The Lagrangian polar dual
$X_{\ell^{\prime}}^{o}\subset\ell^{\prime}$ of $X_{\ell}\subset\ell$ is given
by $X_{\ell^{\prime}}^{o}=S(X^{o})$, that is
\begin{equation}
(X_{\ell},X_{\ell^{\prime}}^{o})=S(X,X^{o}) \label{sxl}%
\end{equation}
\ where $X^{o}$ is the polar dual of $X=S^{-1}(X_{\ell})\subset\ell_{X}$.
\end{proposition}

\begin{proof}
Let us define $P\subset0\times\mathbb{R}_{p}^{n}$ by $(X_{\ell},X_{\ell
^{\prime}}^{o})=S(X,P)$. It suffices to prove that $P=X^{o}$. But this readily
follows from the definition (\ref{ozz}) and the symplectic invariance of the
symplectic form $\omega$.
\end{proof}

This leads us to the following definition:

\begin{definition}
Let $X_{\ell}\subset\ell$ and $Y_{\ell^{\prime}}\subset\ell^{\prime}$ be
centrally symmetric convex bodies. We will say that $(X_{\ell},Y_{\ell
^{\prime}})$ is a polar dual Lagrangian pair if $X_{\ell^{\prime}}^{o}\subset
Y_{\ell^{\prime}}$. If the equality $X_{\ell^{\prime}}^{o}=Y_{\ell^{\prime}}$
holds we will say that $(X_{\ell},Y_{\ell^{\prime}})$ is an exact Lagrangian
polar dual pair.
\end{definition}

Let $S\in\operatorname*{Sp}(n)$ be such that $(\ell,\ell^{\prime})=S(\ell
_{X},\ell_{P})$; in view of Proposition \ref{propsycogeom} $(X_{\ell}%
,Y_{\ell^{\prime}})$ is a dual polar Lagrangian polar pair if and only
$(X,Y)=S^{-1}(X_{\ell},Y_{\ell^{\prime}})$ is a polar dual pair in the usual sense.

The notion of Mahler volume generalizes without difficulty to the Lagrangian
context. Recall that the Mahler volume of a centrally symmetric convex body
$X\subset\mathbb{R}_{x}^{n}$ is defined by
\[
v(X)=\operatorname*{Vol}(X\times X^{o})
\]
where $\operatorname*{Vol}$ is the standard volume in $\mathbb{R}^{2n}$. We
define accordingly
\[
v(X_{\ell})=\operatorname*{Vol}(X_{\ell}\times X_{\ell^{\prime}}^{o}).
\]
Since symplectomorphisms are volume preserving it follows that $v(X_{\ell
})=v(X)$ if $X_{\ell}\times X_{\ell^{\prime}}^{o}=S(X,X^{o}$ ) for
$S\in\operatorname*{Sp}(n)$.

\subsection{L\"{o}wner--John ellipsoids}

Let $\Omega$ be a convex body in $\mathbb{R}^{n}$. The L\"{o}wner and John
ellipsoids $\Omega_{\mathrm{L\ddot{o}wner}}$ and $\Omega_{\mathrm{John}}$ are
defined as follows \cite{ABMB,Ball}:

\begin{itemize}
\item $\Omega_{\mathrm{L\ddot{o}wner}}$ \textit{is the unique ellipsoid in}
$\mathbb{R}^{n}$ \textit{with minimum volume containing} $\Omega$.

\item $\Omega_{\mathrm{Johnr}}$ \textit{is the unique ellipsoid in}
$\mathbb{R}^{n}$ \textit{with maximum volume contained in} $\Omega$;
\end{itemize}

The L\"{o}wner and John ellipsoids are linearly covariant: if $L\in
GL(n,\mathbb{R})$ then
\begin{equation}
(L(\Omega))_{\mathrm{L\ddot{o}wner}}=L(\Omega_{\mathrm{L\ddot{o}wner}})\text{
, }(L(\Omega))_{\mathrm{John}}=L(\Omega_{\mathrm{John}}) \label{JL1}%
\end{equation}
and $\Omega_{\mathrm{L\ddot{o}wner}}$ and $\Omega_{\mathrm{John}}$ are
transformed into each other by polar duality:%
\begin{equation}
\Omega_{\mathrm{L\ddot{o}wner}}=(\Omega^{o})_{\mathrm{John}}\text{ \ ,
\ }\Omega_{\mathrm{John}}=(\Omega^{o})_{\mathrm{L\ddot{o}wner}}. \label{JL2}%
\end{equation}
We will say that $L\in GL(n,\mathbb{R})$ brings $\Omega$ in L\"{o}wner
position (\textit{resp}. John position) if $(L(\Omega))_{\mathrm{L\ddot
{o}wner}}=B^{n}(1)$ (\textit{resp}. $(L(\Omega))_{\mathrm{John}}=B^{n}(1)$).

Replacing $n$ with $2n$ we have the following basic example:

\begin{lemma}
\label{LemLoewner}The John ellipsoid of $B_{X}^{n}(1)\times B_{P}^{n}(1)$ is
$B^{2n}(1)$.
\end{lemma}

\begin{proof}
The inclusion
\begin{equation}
B^{2n}(1)\subset B_{X}^{n}(1)\times B_{P}^{n}(1)\label{incl}%
\end{equation}
is obvious, and we cannot have $B^{2n}(R)\subset B_{X}^{n}(1)\times B_{P}%
^{n}(1)$ if $R>1$. Assume now that the John ellipsoid $\Omega_{\mathrm{John}}$
of $\Omega=B_{X}^{n}(1)\times B_{P}^{n}(1)$ is defined by
\[
Ax\cdot x+Bx\cdot p+Cp\cdot p\leq1
\]
where $A,C>0$ and $B=B^{T}$ are real $n\times n$ matrices. Since $\Omega$ is
invariant by the transformation $(x,p)\longmapsto(p,x)$ so is $\Omega
_{\mathrm{John}}$ and we must thus have $A=C$ and $B=B^{T}$. Similarly,
$\Omega$ being invariant by the reflection $(x,p)\longmapsto(-x,p)$ we get
$B=0$ so $\Omega_{\mathrm{John}}$ is defined by $Ax\cdot x+Ap\cdot p\leq1$.
The next step is to observe that $\Omega$ and hence $\Omega_{\mathrm{John}}$
is invariant under all transformations $(x,p)\longmapsto(Hx,HP)$ where $H\in
O(n,\mathbb{R})$ so we must have $AH=HA$ for all $H\in O(n,\mathbb{R})$, but
this is only possible if $A=\lambda I_{n\times n}$ for some $\lambda
\in\mathbb{R}$. The John ellipsoid is thus of the type $B^{2n}(\lambda
^{-1/2})$ for some $\lambda\geq1$ and this concludes the proof in view of
(\ref{incl}) since $\lambda>1$ is excluded.
\end{proof}

The considerations above allow us to determine the John ellipsoid of $X_{\ell
}\times X_{\ell^{\prime}}^{o}$. Our proof relies on the following
straightforward result:

\begin{lemma}
Let $\Omega\subset\mathbb{R}^{2n}$ be a centrally symmetric body. We have
\begin{equation}
c_{\min}^{\mathrm{lin}}(\Omega)=\sup_{S\in\operatorname*{Sp}(n)}\{\pi
R^{2}:S(B^{2n}(R))\subset\Omega\}~.\label{clinmin}%
\end{equation}

\end{lemma}

\begin{proof}
Since $\Omega$ is centrally symmetric we have $S(B^{2n}(z_{0},R))\subset
\Omega$ if and only if $S(B^{2n}(-z_{0},R))\subset\Omega$. The ellipsoid
$S(B^{2n}(R))$ is interpolated between $S(B^{2n}(z_{0},R))$ and $S(B^{2n}%
(-z_{0},R))$ using the mapping $t\longmapsto$ $z(t)=z-2tz_{0}$ where $z\in
S(B^{2n}(z_{0},R))$, and is hence contained in $\Omega$ by convexity.
\end{proof}

\begin{proposition}
\label{PropJohn}Let $(\ell,\ell^{\prime})$ be a pair of transverse Lagrangian
planes and $X_{\ell}\subset\ell$ a centered ellipsoid. (i) There exists
$S\in\operatorname*{Sp}(n)$ such that the John ellipsoid of $X_{\ell}\times
X_{\ell^{\prime}}^{o}$ is
\[
(X_{\ell}\times X_{\ell^{\prime}}^{o})_{\mathrm{John}}=S(B^{2n}(1)).
\]
(ii) We have
\begin{equation}
c_{\min}^{\mathrm{lin}}(X_{\ell}\times X_{\ell^{\prime}}^{o})=\pi
.\label{cmin3}%
\end{equation}

\end{proposition}

\begin{proof}
(i) In view of Proposition \ref{propsycogeom} there exists $S^{\prime}%
\in\operatorname*{Sp}(n)$ such that
\begin{equation}
(X_{\ell},X_{\ell^{\prime}}^{o})=S^{\prime}(X,X^{o})
\end{equation}
where $X\subset\ell_{X}$ and $X^{o}\subset\ell_{P}$. Using formula (\ref{JL1})
with $L=S$ we have
\[
(X_{\ell}\times X_{\ell^{\prime}}^{o})_{\mathrm{John}}=(S^{\prime}%
(X,X^{o}))_{\mathrm{John}}=S^{\prime}((X,X^{o})_{\mathrm{John}})
\]
so we may assume that there exists $A>0$ such that $X_{\ell}=X$ and
$X_{\ell^{\prime}}^{o}=X^{o}$ with
\begin{align}
X &  =\{x\in\mathbb{R}_{x}^{n}:Ax\cdot x\leq1\}=A^{-1/2}(B_{X}^{n}(1))\\
X^{o} &  =\{p\in\mathbb{R}_{p}^{n}:A^{-1}p\cdot p\leq1\}=A^{1/2}(B_{P}%
^{n}(1)).
\end{align}
This can be rewritten as
\[
(X,X^{o})=S^{\prime\prime}(B_{X}^{n}(1)\times B_{P}^{n}(1))
\]
where
\[
S^{\prime\prime}=%
\begin{pmatrix}
A^{-1/2} & 0\\
0 & A^{1/2}%
\end{pmatrix}
\in\operatorname*{Sp}(n)
\]
and hence
\[
(X_{\ell}\times X_{\ell^{\prime}}^{o})_{\mathrm{John}}=S^{\prime}%
S^{\prime\prime}((B_{X}^{n}(1)\times B_{P}^{n}(1))_{\mathrm{John}}\dot{)}.
\]
The claim now follows from Lemma (\ref{LemLoewner}) taking $S=S^{\prime
}S^{\prime\prime}$. (ii) It is sufficient to assume $X_{\ell}=X\subset\ell
_{X}$ and $X_{\ell^{\prime}}^{o}=X^{o}\subset\ell_{P}$. In view of formula
(\ref{clinmin}) we have to show that
\[
c_{\min}^{\mathrm{lin}}(X_{\ell}\times X_{\ell^{\prime}}^{o})=\sup
_{S\in\operatorname*{Sp}(n)}\{\pi R^{2}:S(B^{2n}(R))\subset X\times X^{o}\}=4.
\]
In view of (i) the largest ellipsoid contained in $X\times X^{o}$ is
$S(B^{2n}(1))$ for some $S\in\operatorname*{Sp}(n)$ hence
\[
c_{\min}^{\mathrm{lin}}(X\times X^{o})=c(B^{2n}(1))=\pi.
\]

\end{proof}

\section{The Main Result: Statement and Proof\label{sec2}}

We begin by recollecting some well-known facts from linear algebra; see for
instance Zhang \cite{zhang}. 

\subsection{Orthogonal projections of ellipsoids}

Consider a non-degenerate ellipsoid
\begin{equation}
\Omega=\{z\in\mathbb{R}^{2n}:Mz\cdot z\leq1\}~. \label{covell1}%
\end{equation}
where $M$ is a real positive definite symmetric $2n\times2n$ matrix (written
for short as $M>0$). We will use the block matrix representation
\begin{equation}
M=%
\begin{pmatrix}
M_{XX} & M_{XP}\\
M_{PX} & M_{PP}%
\end{pmatrix}
\label{M}%
\end{equation}
where the blocks are $n\times n$ matrices, the condition $M>0$ ensures us that
$M_{XX}>0$, $M_{PP}>0$, and $M_{PX}=M_{XP}^{T}$. The $n\times n$ matrices
\begin{align}
M/M_{PP}  &  =M_{XX}-M_{XP}M_{PP}^{-1}M_{PX}\label{schurm1}\\
M/M_{XX}  &  =M_{PP}-M_{PX}M_{XX}^{-1}M_{XP} \label{schurm2}%
\end{align}
are the Schur complements in $M$ of $M_{PP}$ and $M_{XX}$, respectively, and
we have $M/M_{PP}>0$, $M/M_{XX}>0$ (see \cite{zhang}).

\begin{lemma}
\label{LemmaProj}Let $M>0$ and consider the ellipsoid
\[
\Omega=\{z\in\mathbb{R}^{2n}:Mz\cdot z\leq1\}.
\]
The orthogonal projections
\begin{equation}
\Omega_{X}=\Pi_{X}\Omega\text{ \ },\text{ \ }\Omega_{P}=\Pi_{P}\Omega~.
\label{shadows}%
\end{equation}
of $\Omega$ onto $\ell_{X}=\mathbb{R}_{x}^{n}\oplus0$ and $\ell_{P}%
=0\oplus\mathbb{R}_{p}^{n}$, respectively, are the ellipsoids%
\begin{align}
\Omega_{X}  &  =\{x\in\mathbb{R}_{x}^{n}:(M/M_{PP})x^{2}\leq1\} \label{boundb}%
\\
\Omega_{P}  &  =\{p\in\mathbb{R}_{p}^{n}:(M/M_{XX})p^{2}\leq1\}~.
\label{bounda}%
\end{align}

\end{lemma}

\begin{proof}
Let us set $Q(z)=Mz^{2}-1$; the boundary $\partial\Omega$ of the hypersurface
$Q(z)=0$ is defined by%
\begin{equation}
M_{XX}x^{2}+2M_{PX}x\cdot p+M_{PP}p^{2}=1~. \label{mabab}%
\end{equation}
We have $x\in\partial\Omega_{X}$ (the boundary of of $\Omega_{X}$) if and only
if the normal vector to $\partial\Omega$ at the point $z=(x,p)$ is parallel to
$\mathbb{R}_{x}^{n}\times0$ hence we get the constraint $\partial
_{z}Q(z)=2Mz\in\mathbb{R}_{x}^{n}\times0$; this is equivalent to saying that
$M_{PX}x+M_{PP}p=0$, that is to $p=-M_{PP}^{-1}M_{PX}x$. Inserting this value
of $p$ in the equation (\ref{mabab}) shows that $\partial\Omega_{X}$ is the
set of all $x$ such that $(M/M_{PP})x^{2}=1$, which yields (\ref{boundb}).
Formula (\ref{bounda}) is proven in a similar way.
\end{proof}

\subsection{The projection theorem}

Let us state and prove our first main result. Recall that we use the notation
$\ell_{X}=\mathbb{R}_{x}^{n}\oplus0$ \ , \ $\ell_{P}=0\oplus\mathbb{R}_{p}%
^{n}$ for the coordinate Lagrangian planes.

\begin{theorem}
\label{Thm1}Let $\Omega$ be a symmetric convex body in $(\mathbb{R}%
^{2n},\omega)$ containing a symplectic ball $S(B^{2n}(1))$. Let $(\ell
,\ell^{\prime})$ be a pair of transverse Lagrangian planes and denote by
$\Pi_{\ell}:\mathbb{R}^{2n}\longrightarrow\ell$ and $\Pi_{\ell^{\prime}%
}\mathbb{R}^{2n}\longrightarrow\ell^{\prime}$ the oblique projections in the
directions $\ell^{\prime}$ and $\ell$, respectively. Then $(\Pi_{\ell}%
\Omega,\Pi_{\ell^{\prime}}\Omega)$ is a Lagrangian dual pair
\begin{equation}
(\Pi_{\ell}\Omega)_{\ell^{\prime}}^{o}\subset\Pi_{\ell^{\prime}}%
\Omega\label{chz}%
\end{equation}
with equality if and only if $\Omega=S(B^{2n}(1))$.
\end{theorem}

\begin{proof}
\textbf{First case:} we choose $\ell=\ell_{X}$ and $\ell^{\prime}=\ell_{P}$.
Let $S\in\operatorname*{Sp}(n)$ be such that $S(B^{2n}(1))\subset\Omega$ (such
an $S$ is generally not unique). The ellipsoid $S(B^{2n}(1))$ is the set of
all $z\in\mathbb{R}^{2n}$ such that $Mz^{2}\leq1$ where $M=(S^{T})^{-1}S^{-1}$
is a positive definite symplectic matrix; writing $M$ in block-form%
\[
M=%
\begin{pmatrix}
M_{XX} & M_{XP}\\
M_{PX} & M_{PP}%
\end{pmatrix}
\]
we have $M_{XX}>0$, $M_{PP}>0$, $M_{PX}=M_{XP}^{T}$. In view of the projection
Lemma \ref{LemmaProj}
\begin{align}
\Omega_{X} &  =\Pi_{X}\Omega=\{x\in\mathbb{R}_{x}^{n}:(M/M_{PP})x\cdot
x\leq1\}\\
\Omega_{P} &  =\Pi_{P}\Omega=\{p\in\mathbb{R}_{p}^{n}:(M/M_{XX})p\cdot
p\leq1\}
\end{align}
where
\begin{align}
M/M_{PP} &  =M_{XX}-M_{XP}M_{PP}^{-1}M_{PX}\\
M/M_{XX} &  =M_{PP}-M_{PX}M_{XX}^{-1}M_{XP}%
\end{align}
In view of formula (\ref{dualellh})) we have
\[
\Omega_{X}^{o}=\{p:(M/M_{PP})^{-1}p^{2}\leq1\}.
\]
Let us prove that $\Omega_{X}^{o}\subset$ $\Omega_{P}$; the result will then
follow by Lemma \ref{LemmaYQ}. The condition $\Omega_{X}^{o}\subset$
$\Omega_{P}$ is equivalent to
\begin{equation}
(M/M_{PP})^{-1}\geq M/M_{XX}~;\label{RPPXX}%
\end{equation}
let us show that this is the case here. The conditions $M\in\operatorname*{Sp}%
(n)$ and $M=M^{T}$ being equivalent to $MJM=J$ we have the following relations
between the blocks:%
\begin{gather}
M_{XX}M_{PX}=M_{XP}M_{XX}\text{ , \ }M_{PX}M_{PP}=M_{PP}M_{XP}\label{rxx}\\
M_{XX}M_{PP}-M_{XP}^{2}=I_{n\times n}~.\label{rpp}%
\end{gather}
These equalities\ imply that
\begin{align}
M/M_{PP} &  =(M_{XX}M_{PP}-M_{XP}^{2})M_{PP}^{-1}=M_{PP}^{-1}\label{34}\\
M/M_{XX} &  =M_{XX}^{-1}(M_{XX}M_{PP}-M_{XP}^{2})=M_{XX}^{-1}\label{35}%
\end{align}
and hence the inequality (\ref{RPPXX}) holds if and only if $M_{XX}M_{PP}\geq
I_{n\times n}$. Now $M_{XX}M_{PP}=I_{n\times n}+M_{XP}^{2}$ in view of
(\ref{rpp}) hence we will have $M_{XX}M_{PP}\geq I_{n\times n}$ if $M_{XP}%
^{2}\geq0$. To prove that $M_{XP}^{2}\geq0$ it suffices to show that the
eigenvalues of $M_{XP}$ are real. In view of the first formula (\ref{rxx}) we
have $M_{XX}^{1/2}M_{PX}M_{XX}^{-1/2}=M_{XX}^{-1/2}M_{XP}M_{XX}^{1/2}$ hence
$M_{XX}^{-1/2}M_{XP}M_{XX}^{1/2}$ is symmetric and its eigenvalues are real as
claimed. \textbf{General case}\textit{:} Choose $S\in\operatorname*{Sp}(n)$
such that $(\ell,\ell^{\prime})=S(\ell_{X},\ell_{P})$, and consider the
centrally symmetric convex body $S^{-1}(\Omega)$. In view of what has been
proven above the orthogonal projections $\Pi_{X}(S^{-1}(\Omega))$ and $\Pi
_{P}(S^{-1}(\Omega))$ form a dual polar pair in the usual sense. In view of
Proposition \ref{propsycogeom} the sets $S\Pi_{X}(S^{-1}(\Omega))\subset\ell$
and $S\Pi_{P}(S^{-1}(\Omega))\subset\ell^{\prime}$ form a Lagrangian dual
pair. Now $\Pi_{\ell}=S\Pi_{X}S^{-1}$ satisfies $\Pi_{\ell}^{2}=\Pi_{\ell}$
and $\ker(\Pi_{\ell})=S\ell_{P}=\ell^{\prime}$ and similarly \ $\Pi
_{\ell^{\prime}}^{2}=\Pi_{\ell^{\prime}}$ and $\ker(\Pi_{\ell^{\prime}}%
)=S\ell_{X}=\ell$ hence $\Pi_{\ell}$ and $\Pi_{\ell^{\prime}}$ are indeed the
projections described in the statement of the theorem. There remains to study
the case $(\Pi_{\ell}\Omega)_{\ell^{\prime}}^{o}=\Pi_{\ell^{\prime}}\Omega$.
The statement follows from Proposition \ref{PropJohn}. 
\end{proof}

The following consequence is easy:

\begin{corollary}
\label{Lemmaconvex}Let $\Omega$ be a centrally symmetric convex body in
$(\mathbb{R}^{2n},\omega)$ such that $c_{\min}^{\mathrm{lin}}(\Omega)\geq\pi$.
Then $(\Pi_{\ell}\Omega,\Pi_{\ell^{\prime}}\Omega)$ is a Lagrangian dual pair
and the formulas (\ref{chz}) hold.
\end{corollary}

\begin{proof}
The result follows from Theorem \ref{Thm1} in view of formula (\ref{clinmin}%
).  
\end{proof}

\subsection{A partial converse}

Let us address the following problem: given two transverse Lagrangian planes
$\ell$ and $\ell^{\prime}$ in $(\mathbb{R}^{2n},\omega)$ and a Lagrangian
polar dual pair $(X_{\ell},Y_{\ell^{\prime}})$ ($X_{\ell}\subset\ell$ and
$X_{\ell^{\prime}}^{o}\subset Y_{\ell^{\prime}}\subset\ell^{\prime}$), can we
find a (centrally symmetric) convex body $\Omega\subset\mathbb{R}^{2n}$ such
that (with the notation of Theorem \ref{Thm1}) $X_{\ell}=\Pi_{\ell}\Omega$ and
$Y_{\ell^{\prime}}=\Pi_{\ell^{\prime}}\Omega$ ? It is intuitively clear that
the problem has in general infinitely many solutions. To make things more
visible, consider the case $n=1$ where we choose $\ell=\ell_{X}$ (the
\textquotedblleft$x$-axis\textquotedblright) and $\ell=\ell_{P}$ (the
\textquotedblleft$p$-axis\textquotedblright). Choose $X=[-a,a]$ so that
$X^{o}=[-1/a,1/a]$. Any centered ellipse $\Omega$ inscribed in the rectangle
$X\times X^{o}$ (which has area 4) will have orthogonal projections $X$ and
$X^{o}$ on the $x$ and $p$ axes. However, If we require the area of this
ellipse to be $c(\Omega)=\pi$ then the solution is unique, and $\Omega$ is the
ellipse defined by $x^{2}/a^{2}+a^{2}p^{2}\leq1$, which is the John ellipse of
$X\times X^{o}$. This discussion can be extended to the case of arbitrary
dimension $n$:

\begin{theorem}
\label{Thm2}Let $(\ell,\ell^{\prime})$ be a pair of transverse Lagrangian
planes in $(\mathbb{R}^{2n},\omega)$ and $X_{\ell}\subset\ell$ a centered
ellipsoid. Let $X_{\ell^{\prime}}^{o}$ be the Lagrangian dual of $X_{\ell}$.
There exists a unique symplectic ball $\Omega=S(B^{2n}(1))$ ($S\in
\operatorname*{Sp}(n)$) having projections $\Pi_{\ell}\Omega=X_{\ell}$ and
$,\Pi_{\ell^{\prime}}\Omega=X_{\ell^{\prime}}^{o}$.
\end{theorem}

\begin{proof}
\textbf{First case}: Assume that $\ell$ and $\ell^{\prime}$ are the coordinate
planes $\ell_{X}$ and $\ell_{P}$ and that $X_{\ell}$ and $X_{\ell^{\prime}%
}^{o}$ are the ellipsoids%
\begin{align}
X &  =\{x\in\mathbb{R}_{x}^{n}:Ax\cdot x\leq1\}\label{exc1}\\
X^{o} &  =\{p\in\mathbb{R}_{p}^{n}:A^{-1}p\cdot p\leq1\}.\label{exc2}%
\end{align}
Let us show that there exists a symplectic ball $\Omega=S(B^{2n}(1))$ having
orthogonal projections $X$ and $X^{o}$. Setting $M=(SS^{T})^{-1}$ the
ellipsoid $S(B^{2n}(1))$ is the set
\[
\Omega=\{z\in\mathbb{R}^{2n}:Mz\cdot z\leq1\}~.
\]
Let us write $M$ in block-matrix form, as in the proof of Theorem \ref{Thm1}:%
\[
M=%
\begin{pmatrix}
M_{XX} & M_{XP}\\
M_{PX} & M_{PP}%
\end{pmatrix}
.
\]
In view of Lemma \ref{LemmaProj} the orthogonal projections $\Pi_{X}\Omega$
and $\Pi_{P}\Omega$ of $\Omega$ on $\ell_{X}$ and $\ell_{P}$ are the
ellipsoids%
\begin{align}
\Omega_{X} &  =\{x\in\mathbb{R}_{x}^{n}:(M/M_{PP})x\cdot x\leq1\}\\
\Omega_{P} &  =\{p\in\mathbb{R}_{p}^{n}:(M/M_{XX})p\cdot p\leq1\dot{\}}.
\end{align}
The conditions $M\in\operatorname*{Sp}(n)$, $M=M^{T}$ imply that we have
$MJM=J$ hence the following set of relations must be satisfied by the
block-entries of $M$:%
\begin{gather}
M_{XX}M_{PP}-M_{XP}^{2}=I_{n\times n}\label{mc1}\\
M_{XX}M_{PX}=M_{XP}M_{XX}\label{mc2}\\
M_{PX}M_{PP}=M_{PP}M_{XP}~.\label{mc3}%
\end{gather}
Using the identities (\ref{mc2}) and (\ref{mc3}) we get
\begin{align*}
M/M_{PP} &  =M_{XX}-M_{PP}^{-1}M_{PX}^{2}\\
&  =M_{PP}^{-1}(M_{PP}M_{XX}-M_{PX}^{2})\\
&  =M_{PP}^{-1}%
\end{align*}
the last equality being obtained by transposition of the identity (\ref{mc1}).
A similar calculation leads to the formula%
\[
M/M_{XX}=M_{PP}-M_{XX}^{-1}M_{XP}^{2}=M_{XX}^{-1}%
\]
and comparing with (\ref{exc1}) and (\ref{exc2}) we must thus have
$A=M_{XX}=M_{PP}^{-1}$. From the identity (\ref{mc1}) follows that $M_{XP}%
^{2}=0$; multiplying both sides of the identity (\ref{mc2}) on the left by
$M_{XP}$ then implies that $M_{PX}^{T}M_{XX}M_{PX}=M_{XP}M_{XX}M_{PX}=0$ and
hence $M_{XP}=0$ since $M_{XX}$ is invertible. The symplectic matrix $M$ must
thus be the block diagonal matrix
\begin{equation}
M=%
\begin{pmatrix}
A & 0\\
0 & A^{-1}%
\end{pmatrix}
\label{MA}%
\end{equation}
and $\Omega$ is hence the ellipsoid
\begin{equation}
\Omega=\{(x,p)\in\mathbb{R}^{n}\times\mathbb{R}^{n}:Ax\cdot x+A^{-1}%
p\cdot\cdot p\leq1\}.\label{omega1}%
\end{equation}
The latter is the John ellipsoid of $(X,X^{o})$ in view of Lemma
\ref{LemLoewner}. \textbf{Second case}: It is similar to the proof given at
the end of Theorem \ref{Thm1}: choose $S^{\prime}\in\operatorname*{Sp}(n)$
such that $(\ell,\ell^{\prime})=S^{\prime}(\ell_{X},\ell_{P})$ and consider
the ellipsoid
\[
\Omega^{\prime}=S^{\prime}\Omega=S^{\prime}S(B^{2n}(1)))
\]
($S$ and $\Omega$ as above). The projections $\Pi_{\ell}$ and $,\Pi
_{\ell^{\prime}}$ are given by $\Pi_{\ell}=S^{\prime}\Pi_{X}S^{\prime-1}$ and
$\Pi_{\ell^{\prime}}=S^{\prime}\Pi_{P}S^{\prime-1}$ hence%
\[
(\Pi_{\ell}\Omega^{\prime},\Pi_{\ell^{\prime}}\Omega^{\prime})=S^{\prime}%
(\Pi_{X}\Omega,\Pi_{P}\Omega)
\]
is a Lagrangian dual pair. The uniqueness of $\Omega^{\prime}$ follows using
the linear covariance (\ref{JL1}) of the John ellipsoid. Formula (\ref{min1})
is obvious by definition of $c_{\min}^{\mathrm{lin}}$ and the fact that if
$S(B^{2n}(R))\subset(X\times X^{o})$ for some $S\in\operatorname*{Sp}(n)$ then
we must have $R\leq1$ by definition of the John ellipsoid.
\end{proof}

\begin{remark}
The fact that $X$ and $X^{o}$ are the orthogonal projections on the coordinate
Lagrangian planes of the ellipsoid (\ref{omega1}) is obvious and might lead to
think that the rather long argument developed in the proof is superfluous;
what is not obvious in the multi-dimensional case is the uniqueness of an
ellipsoid having this property.
\end{remark}

\section{The Uncertainty Principle of Quantum Mechanics\label{secup}}

\subsection{Uncertainties}

The term \textquotedblleft uncertainty principle\textquotedblright\ commonly
hints in mathematics at a constellation of related results based on the
observation of that there is a kind of \textquotedblleft
trade-off\textquotedblright\ between a function $\psi$ and its Fourier
transform $F\psi$ which prevents them to be simultaneously too sharply
located. Perhaps the most emblematic of these results (and one of the oldest)
is Hardy's \cite{ha32} uncertainty principle which says that if $\psi
(x)=\mathcal{O}(e^{-ax^{2}/2})$ and $F\psi(p)=\mathcal{O}(e^{-ap^{2}/2})$ as
$x,p\rightarrow\infty$ then we must have $ab\leq1$, with equality if and only
if $\psi(x)=\mathcal{C}e^{-ax^{2}/2}$ for some constant $C$. In practice, the
best known uncertainty principles are those arising from quantum mechanics,
and are variations on the theme of the Heisenberg principle of indeterminacy
$\sigma_{xx}\sigma_{pp}\geq\frac{1}{4}\hbar^{2}$. The most used of these
results is, with no doubt, the Robertson--Schr\"{o}dinger inequality%
\[
\sigma_{xx}\sigma_{pp}\geq\sigma_{xp}^{2}+\tfrac{1}{4}\hbar^{2}%
\]
where $\sigma_{xx}$ and $\sigma_{pp}$ are variances (also called standard
deviations) and $\sigma_{xp}$ is the covariance corresponding to joint
position and momentum measurements. As we already observed in \cite{gopolar}
the rub with these types of inequalities, generalizing Heisenberg's principle,
is that they crucially depend on the choice of a privileged (and arbitrary)
way of measuring uncertainties, in his case the usual (co-)variances from
standard statistics. In fact, Hilgevoord \cite{hi02} already noticed some time
ago (also see the follow-up \cite{hiuf} with Uffink), that variances and
covariances give good measurements of the spread only for Gaussian (or almost
Gaussian) quantum states, this being due to the fact they usually fail to take
into account the \textquotedblleft tails\textquotedblright\ of $\psi$ which
can be quite large and influence the variance and covariance.

We are going to see that the notion of dual polarity is a very intuitive and
general tool for expressing uncertainties, which avoids the pitfalls due to
the use of standard statistical tools. We will thereafter show that Hardy's
uncertainty principle can also be interpreted in a convincing way using polar duality.

We will from now on use units in which $\hslash=1$.

\subsection{The Robertson--Schr\"{o}dinger inequalities}

Let $\mathcal{B}=\mathcal{B}(L^{2}(\mathbb{R}^{n}))$ be the algebra of bounded
operators on $L^{2}(\mathbb{R}^{n})$. We denote by $\mathcal{L}_{1}%
=\mathcal{L}_{1}(L^{2}(\mathbb{R}^{n}))$ the two-sided ideal of $\mathcal{B}$
consisting of all trace-class operators on $L^{2}(\mathbb{R}^{n})$. An
operator $\widehat{\varrho}\in\mathcal{B}(L^{2}(\mathbb{R}^{n}))$ is called a
\textit{density operator} if it is positive semi-definite: $\widehat{\varrho
}\geq0$ (and hence self-adjoint) and if it has trace $\operatorname*{Tr}%
(\widehat{\varrho})=1$. It follows from the spectral theorem that there exists
an orthonormal sequence of vectors $(\psi_{j})_{j\in\mathbb{N}}$ in
$L^{2}(\mathbb{R}^{n})$ and a corresponding sequence $(\lambda_{j}%
)_{j\in\mathbb{N}}$ of non-negative real numbers summing up to one such that
for every $\psi\in L^{2}(\mathbb{R}^{n})$%
\begin{equation}
\widehat{\varrho}\psi=\sum_{j\in\mathbb{N}}\lambda_{j}(\psi|\psi_{j})_{L^{2}%
}\psi_{j}. \label{spectral}%
\end{equation}
One proves \cite{Birkbis} that the Weyl symbol of $\widehat{\varrho}$ is the
function (\textquotedblleft Wigner distribution\textquotedblright)
\begin{equation}
\rho=(2\pi)^{n}\sum_{j\in\mathbb{N}}\lambda_{j}W\psi_{j} \label{specwig}%
\end{equation}
where $W\psi_{j}$ is the Wigner transform of $\psi_{j}$; recall
\cite{Birkbis,Wigner} that for $\psi\in L^{2}(\mathbb{R}^{n})$ the function
$W\psi$ is given by the absolutely convergent integral%
\begin{equation}
W\psi(x,p)=(2\pi)^{-n}\int_{\mathbb{R}^{n}}e^{-ip\cdot y}\psi(x+\tfrac{1}%
{2}y)\overline{\psi(x-\tfrac{1}{2}y)}dy. \label{Wig}%
\end{equation}
In order to define the covariance matrix of $\widehat{\varrho}$ we need to
consider density operators whose Wigner distribution decreases sufficiently
well at infinity: \cite{Birkbis,fe06,Jakobsen,Gro}:

\begin{definition}
We will say that $\widehat{\varrho}$ has the Feichtinger property for $s\geq0$
if the functions $\psi_{j}$ in the spectral decomposition (\ref{spectral}) of
$A$ satisfy
\begin{equation}
W(\psi_{j},\phi)\in L_{s}^{1}\mathbb{R}^{2n})\text{ \ },\text{ \ }%
j\in\mathbb{N} \label{7}%
\end{equation}
for one (and hence every) every $\phi\in\mathcal{S}(\mathbb{R}^{n})$.
\end{definition}

In formula (\ref{7}) $W(\psi,\phi)$ is the cross-Wigner transform, defined by
the integral
\begin{equation}
W(\psi,\phi)(z)=\left(  \tfrac{1}{2\pi}\right)  ^{n}\int_{\mathbb{R}^{n}%
}e^{-ip\cdot y}\psi(x+\tfrac{1}{2}y)\overline{\phi(x-\tfrac{1}{2}y)}dy
\end{equation}
and $L_{s}^{1}\mathbb{R}^{2n})$ is the weighted $L^{1}$-space%
\begin{equation}
L_{s}^{1}(\mathbb{R}^{2n})=\{a:\mathbb{R}^{2n}\longrightarrow\mathbb{C}\text{
},\text{ }\langle z\rangle^{s}a\in L^{1}(\mathbb{R}^{2n})\} \label{6}%
\end{equation}
where $\langle z\rangle=(1+|z|^{2})^{1/2}$. It is easy to show \cite{CMQR}
that if $\rho$ has the Feichtinger property for some $s\geq2$ then $\rho\in
L^{1}(\mathbb{R}^{2n})$ and the symmetric $2n\times2n$ matrix
\begin{equation}
\Sigma=\int\nolimits_{\mathbb{R}^{2n}}(z-\overline{z})(z-\overline{z}%
)^{T}a(z)dz \label{vectorform}%
\end{equation}
where $\overline{z}=\int\nolimits_{\mathbb{R}^{2n}}z\rho(z)dz$. is defined
(both integrals being absolutely convergent). This matrix $\Sigma$ is called
the covariance matrix of the density operator $\widehat{\varrho}$, and it is
well-known \cite{Narcow1,Birkbis} that the positive semi-definiteness of
$\widehat{\varrho}$ implies the condition%
\begin{equation}
\Sigma+\frac{i}{2}J\geq0. \label{quantum}%
\end{equation}
This property, which is related to the so-called KLM conditions
\cite{LouMiracle1,LouMiracle2}, has the following converse: let $\Sigma$ be a
real symmetric $2n\times2n$ matrix satisfying (\ref{quantum}). Then, for every
$\overline{z}\in\mathbb{R}^{2n}$, the Gaussian probability distribution
\begin{equation}
\rho(z)=\frac{1}{(2\pi)^{n}\sqrt{\det\Sigma}}e^{-\frac{1}{2}\Sigma
^{-1}(z-\overline{z})^{2}} \label{Gauss}%
\end{equation}
is the Wigner distribution of a density operator $\widehat{\varrho}$ (and
$\widehat{\varrho}$ has trivially the Feichtinger property for every $s\geq
0$). Writing the covariance matrix in block form%
\begin{equation}
\Sigma=%
\begin{pmatrix}
\Sigma_{XX} & \Sigma_{XP}\\
\Sigma_{PX} & \Sigma_{PP}%
\end{pmatrix}
\text{ \ },\text{ \ }\Sigma_{PX}=\Sigma_{XP}^{T} \label{defcovma}%
\end{equation}
with $\Sigma_{XX}=(\sigma_{x_{j}x_{k}})_{1\leq j,k\leq n}$, $\Sigma
_{PP}=(\sigma_{p_{j}p_{k}})_{1\leq j,k\leq n}$, and $\Sigma_{XP}%
=(\sigma_{x_{j}p_{k}})_{1\leq j,k\leq n}$ it follows from Sylvester's
criterion for the leading principal minors of a positive matrix that the
condition (\ref{quantum}) is equivalent to the Robertson--Schr\"{o}dinger
(RS)\ inequalities
\begin{equation}
\sigma_{x_{j}x_{j}}\sigma_{p_{j}p_{j}}\geq\sigma_{x_{j}p_{j}}^{2}+\tfrac{1}%
{4}\text{ \ },\text{ \ }1\leq j\leq n, \label{RS}%
\end{equation}
for position and momentum measurements (\cite{go09,Birkbis}.

As discussed in previous section, this formulation of the uncertainty
principle is heavily dependent on the way uncertainties are measured (in this
case the variances and covariances of certain stochastic variables). We are
going to show that polar duality allows to formulate a more general principle.

\subsection{Polar duality and a conjecture}

Suppose that a large number of position and momentum measurements are made on
some physical system, identified with a point in $T^{\ast}\mathbb{R}%
^{n}=\mathbb{R}_{x}^{n}\times\mathbb{R}_{p}^{n}\equiv\mathbb{R}^{2n}$. These
measurements yield \textquotedblleft clouds\textquotedblright\ of points in
$\mathbb{R}_{x}^{n}$ and $\mathbb{R}_{p}^{n}$, respectively. Using standard
methods from statistics and optimization theory \cite{ALG2,ALG1} we can model
these subsets by \textquotedblleft best fit\textquotedblright\ ellipsoids
$\Omega_{X}\subset\mathbb{R}_{x}^{n}$ and $\Omega_{X}\subset\mathbb{R}_{p}%
^{n}$ (which we assume, for simplicity, both centered at the origin). For
instance, a standard and robust statistical method is to use Van Aelst and
Rousseuw's \cite{MVE} Minimum Volume Ellipsoid method (MVE) ,of which we have
given a short description in \cite{go09}; also see our paper \cite{stat} where
the MVE method is related to the notion of symplectic capacity. Now, according
to the principle of quantum indeterminacy \cite{neu}, it does not make sense
to attribute a precise value to the position or the momentum, so the sets of
measurements, ad hence the ellipsoids $\Omega_{X}$ and $\Omega_{P}$, cannot be
simultaneously arbitrarily small. This leads us to make the following physical
conjecture, which could be (in principle) justified experimentally:

\begin{conjecture}
The ellipsoids $\Omega_{X}\subset\mathbb{R}_{x}^{n}$ and $\Omega_{P}%
\subset\mathbb{R}_{p}^{n}$ form a dual pair: $\Omega_{X}^{o}\subset\Omega_{P}$.
\end{conjecture}

We are going to make the above conjecture plausible by showing that it implies
(but is more general then) the RS inequalities (\ref{RS}). We begin by
defining two real symmetric $n\times n$ ellipsoids by identifying the
ellipsoids $\Omega_{X}$ and $\Omega_{P}$ with the sets
\begin{align*}
\Omega_{X}  &  =\{x\in\mathbb{R}_{x}^{n}:\tfrac{1}{2}\Sigma_{XX}^{-1}x\cdot
x\leq1\}\\
\Omega_{X}  &  =\{p\in\mathbb{R}_{p}^{n}:\tfrac{1}{2}\Sigma_{PP}^{-1}p\cdot
p\leq1\}.
\end{align*}
Let us assume that $(\Omega_{X},\Omega_{P})$ is an exact polar dual pair, that
is $\Omega_{X}^{o}=\Omega_{P}$. In view of the duality property
(\ref{dualellh}) for ellipsoids we must then have $\Sigma_{XX}\Sigma
_{PP}=\frac{1}{4}I_{n\times n}$ (which reduces to the equality $\sigma
_{xx}\sigma_{pp}=\frac{1}{4}$ for $n=1$). It follows that we have%
\[
\frac{1}{2}\Sigma^{-1}=%
\begin{pmatrix}
\frac{1}{2}\Sigma_{XX}^{-1} & 0\\
0 & \frac{1}{2}\Sigma_{PP}^{-1}%
\end{pmatrix}
\in\operatorname*{Sp}(n)
\]
hence $\Sigma$ automatically satisfies the condition (\ref{quantum}) so the RS
inequalities (\ref{RS}) are satisfied, in fact we have $\sigma_{x_{j}x_{j}%
}\sigma_{p_{j}p_{j}}\geq\tfrac{1}{4}$ for $1\leq j\leq n$. Notice that the
Gaussian distribution (\ref{Gauss}) is here the Wigner transform of the
function%
\begin{equation}
\psi(x)=\left(  \tfrac{1}{2\pi}\right)  ^{n/4}(\det\Sigma_{XX})^{-1/4}%
e^{-\tfrac{1}{4}\Sigma_{XX}^{-1}x\cdot x}. \label{psix}%
\end{equation}

This discussion generalizes to arbitrary dual polar pairs; the key to this
generalization is the following result:

\begin{lemma}
Let $\Sigma>0$. The condition $\Sigma+\frac{1}{2}iJ\geq0$ is equivalent to the
existence of $S\in\operatorname*{Sp}(n)$ such that $S(B^{2n}(1))\subset\Omega$
(that is to the condition $c(\Omega)\geq\pi$) where $\Omega$ is the ellipsoid
\[
\Omega=\{z\in\mathbb{R}^{2n}:\tfrac{1}{2}\Sigma^{-1}z\cdot z\leq1\}.
\]

\end{lemma}

\begin{proof}
See \cite{go09,Birkbis} for detailed proofs.
\end{proof}

As we have discussed in \cite{go09} the condition $c(\Omega)\geq\pi$ thus
implies the RS inequalities (\ref{RS}), and can therefore be viewed as a
symplectically invariant expression of the uncertainty principle of quantum mechanics.

\begin{proposition}
\label{PropUP}Let $(\ell,\ell^{\prime})$ be a pair of transverse Lagrangian
planes and $\Omega_{X}\subset\ell$, $\Omega_{P}\subset\ell^{\prime}$ two
centered ellipsoids. (i) If $(\Omega_{X},\Omega_{P})$ is a Lagrangian dual
pair, there exists ellipsoids $\Omega$ such that $\Pi_{\ell}\Omega=\Omega_{X}$
and $\Pi_{\ell^{\prime}}\Omega=\Omega_{P}$ and $\Omega$ has symplectic
capacity $c(\Omega)\geq\pi$; (ii) Equivalently, the matrix $\Sigma$ defined
by
\[
\Omega=\{z\in\mathbb{R}^{2n}:\tfrac{1}{2}\Sigma^{-1}z\cdot z\leq1\}
\]
is the covariance matrix of a density operator with Gaussian Wigner
distribution%
\begin{equation}
\rho(z)=\frac{1}{(2\pi)^{n}\sqrt{\det\Sigma}}e^{-\frac{1}{2}\Sigma
^{-1}(z-\overline{z})^{2}}. \label{Gaussbis}%
\end{equation}
(iii) When $(\Omega_{X},\Omega_{P})$ is an exact pair, i.e. when $\Omega
_{X}^{o}=\Omega_{P}$, then $\rho(z)=W\psi(z)$ where $\psi$ is a Gaussian
\begin{equation}
\psi(x)=\left(  \tfrac{1}{\pi\hbar}\right)  ^{n/4}(\det A)^{1/4}e^{-\tfrac
{1}{2\hbar}(A+iB)x^{2}} \label{fay}%
\end{equation}
where $A$ and $B$ are symmetric real $n\times n$ matrices and $A>0$.
\end{proposition}

\begin{proof}
Suppose first that $\Omega_{X}=\Omega_{X}^{o}$. In view of Theorem
\ref{Thm2}(i) there exists a unique symplectic ball $\Omega=S(B^{2n}(1))$
($S\in\operatorname*{Sp}(n)$) having projections $\Pi_{\ell}\Omega=X_{\ell}$
and $,\Pi_{\ell^{\prime}}\Omega=X_{\ell^{\prime}}^{o}$ and this ellipsoid is
the John ellipsoid of $X_{\ell}\times X_{\ell^{\prime}}^{o}$.
\end{proof}

\begin{remark}
The result above is closely related to \textquotedblleft Pauli's
reconstruction problem\textquotedblright\ \cite{Pauli,espo} where one wants to
recover a function $\psi$ knowing its modulus and the modulus of its Fourier
transform; see our discussion of Pauli's problem in \cite{gopolar}.
\end{remark}

\section{Hardy's Uncertainty Principle\label{sechardy}}

We briefly discussed Hardy's uncertainty principle in the beginning of Section
\ref{secup}. Let us study the multi-dimensional variant of it from the point
of view of polar duality; we will extend previous results in \cite{ACHA}. The
Fourier transform of $\psi\in L^{1}(\mathbb{R}^{n})\cap L^{2}(\mathbb{R}^{n})$
is defined by%
\begin{equation}
F\psi(p)=\left(  \tfrac{1}{2\pi}\right)  ^{n/2}\int_{\mathbb{R}^{n}}%
e^{-ipx}\psi(x)dx \label{fh}%
\end{equation}
and extends extends into a unitary automorphism $L^{2}(\mathbb{R}%
^{n})\longrightarrow L^{2}(\mathbb{R}^{n})$ also denoted by $F$. In
\cite{golu07} we have proven the following generalization Hardy's result:
assume that $\psi\in L^{2}(\mathbb{R}^{n})$, $||\psi||\neq0$, satisfies the
sub-Gaussian estimates%
\begin{equation}
|\psi(x)|\leq Ce^{-\tfrac{1}{2}Ax\cdot x}\text{ \ and \ }|F\psi(p)|\leq
Ce^{-\tfrac{1}{2}Bp\cdot p} \label{Hardy2}%
\end{equation}
where $A$ and $B$ are two real positive definite symmetric matrices and $C>0$. Then:

\begin{proposition}
\label{propAB}(i) The eigenvalues $\lambda_{j}$, $j=1,...,n$, of $AB$ are
$\leq1$; (ii) If $\lambda_{j}=1$ for all $j$, then $\psi(x)=Ce^{-\frac{1}%
{2}Ax^{2}}$ for some complex constant $C$.
\end{proposition}

The proof of this result is based on the following diagonalization result
which is a refinement of Williamson's symplectic diagonalization theorem in
the block-diagonal case:

\begin{lemma}
\label{lem1}Let $A$ and $B$ be two real positive-definite symmetric $n\times
n$ matrices. There exists $L\in GL(n,\mathbb{R})$ such that%
\begin{equation}%
\begin{pmatrix}
L^{T} & 0\\
0 & L^{-1}%
\end{pmatrix}%
\begin{pmatrix}
A & 0\\
0 & B
\end{pmatrix}%
\begin{pmatrix}
L & 0\\
0 & (L^{T})^{-1}%
\end{pmatrix}
=%
\begin{pmatrix}
\Lambda & 0\\
0 & \Lambda
\end{pmatrix}
\label{lalb}%
\end{equation}
\ where $\Lambda=\operatorname*{diag}(\sqrt{\lambda_{1}},...,\sqrt{\lambda
_{n}})$ is the diagonal matrix whose eigenvalues are the square roots of the
eigenvalues $\lambda_{1},...,\lambda_{n}$ of $AB$.
\end{lemma}

Proposition \ref{propAB} implies the following geometric version of Hardy's
uncertainty principle:

\begin{corollary}
\label{cor1}Let $X_{A}=\{x\in\mathbb{R}^{n}:Ax\cdot x\leq1\}$ and
$P_{B}=\{p\in\mathbb{R}^{n}:Bp\cdot p\leq1\}$. (i) There exists a $\psi\in
L^{2}(\mathbb{R}^{n})$ with $||\psi||\neq0$ satisfying the multi-dimensional
Hardy inequalities (\ref{Hardy2}) if and only if $(X_{A},P_{B})$ is a dual
pair, (ii) we have $X_{A}^{o}=P_{B}$ if and only if $\psi(x)=Ce^{-\frac{1}%
{2}Ax^{2}}$ for some constant $C>0$.
\end{corollary}

\begin{proof}
(i) immediately follows from Proposition \ref{propAB} using Proposition
\ref{propell}. (ii) We have $X_{A}^{o}=P_{B}$ if and only if all the
eigenvalues of $AB$ are all equal to one hence the result by Proposition
\ref{propAB}.
\end{proof}

\subsection{Sub-Gaussian estimates for the Wigner transform}

Assuming that $\psi$ and its Fourier transform $F\psi$ are in $L^{1}%
(\mathbb{R}^{n})\cap L^{2}(\mathbb{R}^{n})$ the marginal conditions
\begin{align}
\int_{\mathbb{R}^{n}}W\psi(z)dp  &  =|\psi(x)|^{2}\label{marginal1}\\
\int_{\mathbb{R}^{n}}W\psi(z)dx  &  =|F\psi(p)|^{2} \label{marginal2}%
\end{align}
hold \cite{Birk,Wigner}. Since the knowledge of $W\psi$ determines $\psi$ up
to a constant unimodular factor, Hardy's uncertainty principle can be
reformulated in terms of the Wigner transform $W\psi$.

We apply the multi-dimensional Hardy uncertainty principle to the study of
estimates of the type%
\begin{equation}
W\psi(z)\leq Ce^{-Mz\cdot z} \label{wighardy}%
\end{equation}
where $M>0$ (\textit{i.e. }$M$ is real positive-definite and symmetric), $C>0$
and $\psi\in L^{2}(\mathbb{R}^{n})$. We will in fact assume that,in addition,
$\psi$ and its Fourier transform $F\psi$ are in $L^{2}(\mathbb{R}^{n})$ so
that the marginal conditions (\ref{marginal1}) and (\ref{marginal2})hold.

We have the following result which can be viewed as an extension of Corollary
\ref{cor1}:

\begin{proposition}
Assume that there exists a non-zero $\psi\in L^{2}(\mathbb{R}^{n})$ such that
the inequality%
\[
W\psi(z)\leq Ce^{-Mz\cdot z}\text{ \ , \ }z\in\mathbb{R}^{2n}%
\]
holds for some $C>0$. Let
\[
\Omega=\{z\in\mathbb{R}^{2n}:Mz\cdot z\leq1\}.
\]
(i) For every pair $(\ell,\ell^{\prime})$ of transversal Lagrangian planes
$(\Pi_{\ell}\Omega,\Pi_{\ell^{\prime}}\Omega)$ is a dual Lagrangian dual pair.
(ii) If $(\Pi_{\ell}\Omega)_{\ell^{\prime}}^{o}=\Pi_{\ell^{\prime}}\Omega$
then $W\psi$ is a Gaussian $Ce^{-Mz\cdot z}$.
\end{proposition}

\begin{proof}
(i) In view of Theorem \ref{Thm1} it is again sufficient to prove that
$c(\Omega)\geq\pi$ for any symplectic capacity $c$. In view of Williamson's
symplectic diagonalization theorem \cite{Birk,HZ} there exists $S\in
\operatorname*{Sp}(n)$ such that
\[
M=S^{T}DS\text{ , }D=%
\begin{pmatrix}
\Lambda & 0\\
0 & \Lambda
\end{pmatrix}
\]
where $\Lambda$ is the diagonal matrix with non-zero entries the symplectic
eigenvalues of $M$. Using the symplectic covariance of the Wigner distribution
\cite{Birk} it follows that we have
\begin{equation}
W(\widehat{S}\psi)(z)=W\psi(S^{-1}z)\leq Ce^{-Dz\cdot z} \label{1}%
\end{equation}
where $\widehat{S}\in\operatorname*{Mp}(n)$ is any of the two metaplectic
operators covering $S$. Setting $\psi^{\prime}=\widehat{S}\psi$ and
integrating this inequality with respect to $p$ and then $x$ we get, using the
marginal conditions (\ref{marginal1}) and (\ref{marginal2}),
\begin{equation}
|\psi^{\prime}(x)|^{2}\leq C^{\prime}e^{-Dx\cdot x}\text{ , }|F\psi^{\prime
}(p)|^{2}\leq C^{\prime}e^{-Dp\cdot p} \label{2}%
\end{equation}
where $C^{\prime}>0$ is a constant depending only on $C$ and the symplectic
eigenvalues $\lambda_{j}^{\omega}$ of $M$. Applying the multidimensional Hardy
uncertainty principle, we must have $\lambda_{j}^{\omega}\leq1$ for $1\leq
j\leq n$ and hence $c(\Omega)\geq\pi$ in view of formula (\ref{capellipsoid}).
The result now follows from Theorem \ref{Thm1}.
\end{proof}

\begin{acknowledgement}
This work has been financed by the Grant P 33447 N of the Austrian Research
Foundation FWF.
\end{acknowledgement}

\end{document}